\documentclass[a4paper,UKenglish,cleveref, autoref, thm-restate, nolineno]{socg-lipics-v2021}

\pdfoutput=1 
\hideLIPIcs  

\newcommand{\newR}{\tilde{r}}
\newcommand{\col}{\texttt{col}}
\DeclareMathOperator{\rank}{\texttt{rank}}
\DeclareMathOperator{\select}{\texttt{select}}
\DeclareMathOperator{\access}{\texttt{access}}

\DeclareMathOperator{\succe}{\texttt{succ}}

\newcommand{\width}{\texttt{m}}
\newcommand{\height}{\texttt{h}}
\newcommand{\word}{w}

\newcommand{\PA}{\texttt{PA}}
\newcommand{\PBWT}{\texttt{PBWT}}
\newcommand{\SA}{\texttt{SA}}

\newcommand{\HList}{\texttt{RS}}

\title{Faster Iterative $\phi$ Queries on the Positional BWT} 

\author{Paola Bonizzoni}{Department of Computer Science, University of Milano-Bicocca, Italy}{paola.bonizzoni@unimib.it}{https://orcid.org/0000-0001-7289-4988}{}

\author{Travis Gagie}{Department of Computer Science, Dalhousie University, Canada}{travis.gagie@gmail.com}{https://orcid.org/0000-0003-3689-327X}{}

\author{Younan Gao}{Department of Computer Science, University of Milano-Bicocca, Italy}{younan.gao@unimib.it}{https://orcid.org/0000-0003-4984-2551}{}

\authorrunning{Bonizzoni et al.} 

\Copyright{Paola Bonizzoni, Travis Gagie, Younan Gao} 



\category{} 

\relatedversion{} 

\funding{Travis Gagie is funded by NSERC Discovery Grant RGPIN-07185-2020. Paola Bonizzoni and Younan Gao have received funding from the European Union’s Horizon 2020 research and innovation programme under the Marie Sk{\l}odowska-Curie grant agreement PANGAIA No.~872539, as well as from grant MIUR 2022YRB97K (PINC, \emph{Pangenome Informatics: From Theory to Applications}), funded by the European Union under the NextGenerationEU programme, Mission~4.}

\nolinenumbers 

\EventEditors{John Q. Open and Joan R. Access}
\EventNoEds{2}
\EventLongTitle{42nd Conference on Very Important Topics (CVIT 2016)}
\EventShortTitle{CVIT 2016}
\EventAcronym{CVIT}
\EventYear{2016}
\EventDate{December 24--27, 2016}
\EventLocation{Little Whinging, United Kingdom}
\EventLogo{}
\SeriesVolume{42}
\ArticleNo{23}

\begin{document}

\maketitle

\begin{abstract}
The Positional Burrows–Wheeler Transform (PBWT) is a fundamental data structure for the efficient representation and analysis of large-scale haplotype panels. For a panel of $\height$ sequences $\{S_1, \dots, S_{\height}\}$ over $\width$ sites, a key operation is the $\phi_j(i)$ query, which returns the haplotype index immediately preceding $S_i$ in co-lexicographic order at site $j$. Efficient support for $k$ iterative queries $\phi^1, \dots, \phi^{k}$ is essential for haplotype matching and variation analysis.

In this work, we introduce a simple and novel decomposition scheme that decomposes each haplotype row into  sub-intervals, called \emph{refined segments},   within which a haplotype’s co-lexicographic predecessor for the sites   remains unchanged. We show that refined segments satisfy two key properties: (i) each segment $[b, e]$ associated with $S_i$ overlaps with at most a constant number of segments of $S_{\phi_e(i)}$, and (ii) the total number of segments is bounded by $O(\newR + \height)$, where $\newR$ denotes the number of runs in the PBWT.
Building on this decomposition, we present two space–time tradeoffs for supporting $k$ iterative $\phi$ queries:
(i) a structure using $O((\newR + \height)\log n)$ bits of space that answers $k$ iterative queries in $O(\log \log_{\word} \min(\width, \height) + k)$ time, where $n = \width \cdot \height$, and (ii) a more compact structure using $O(\newR \log \height + \height \log n)$ bits of space that supports queries in $O(k \log \log_{\word} \height)$ time.

Prior to our work, supporting these queries required $O((\newR + \height)\log n)$ bits of space and $O(k \cdot \log \log_{\word} \width)$ time. Our second tradeoff is expected to be effective in practice for modern genomic datasets, where the number $\height$ of haplotypes is typically much smaller than the number $\width$ of sites.
\end{abstract}

\section{Introduction}

The management of large-scale genomic datasets is a central challenge in bioinformatics. A typical component of these resources is the haplotype panel, which captures genetic variation across a population. Such panels consist of $\height$ sequences (haplotypes) observed at $\width$ variant sites and are conventionally represented as an $\height \times \width$ matrix, where rows correspond to haplotypes and columns to genomic positions. 
Haplotype panels are used to support several common genomic applications, including the computation of \emph{matching statistics}~\cite{Cozzi2023-iv, Bonizzoni2024-dj}, the identification of \emph{minimal positional substring covers} (MPSC)~\cite{Bonizzoni2024-dj}, and the search for \emph{set-maximal exact matches} (SMEMs)~\cite{Cozzi2023-iv, Bonizzoni2024-dj}.

To address the computational demands of these applications, Richard Durbin introduced the Positional Burrows–Wheeler Transform (PBWT)~\cite{Durbin2014-dd}. The PBWT is a data structure inspired by suffix array techniques~\cite{ManberMyers1993} and comprises two main components: the prefix array (PA) and the PBWT matrix itself.
The PA and the PBWT are $\height \times \width$ matrices in which each column $j$ stores, respectively, the permutation of haplotype indices $\{1, \dots, \height\}$ and the permutation of the $j$-th column of the haplotype panel induced by the co-lexicographical ordering of haplotype prefixes up to column $j-1$.

\begin{figure}[t]
    \centering
    \includegraphics[width=1\textwidth]{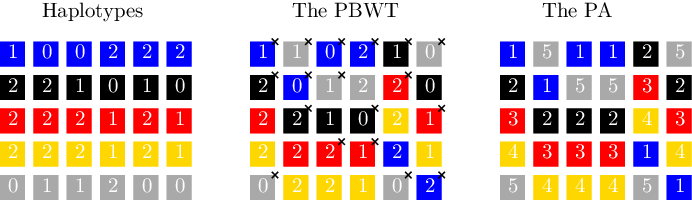}
    \caption{Example of PBWT and Prefix array. Crosses mark run-tops in the PBWT. The haplotypes $S_1, S_2, S_3, S_4, S_5$ are assigned the colors $\{\text{blue, black, red, yellow, gray}\}$, respectively.}
    \label{fig:hap-pbwt-pa}
\end{figure}

Querying the PBWT and the PA relies on three fundamental operations. \emph{Forward stepping} maps the position of a haplotype in the permutation at site $j$ to its position at site $j{+}1$, while \emph{backward stepping} performs the inverse mapping. A $\phi_j(i)$ query returns the index immediately preceding $i$ in the $j$-th column of the PA.
These three operations together serve as the core primitives for the fundamental haplotype panel queries described earlier.

Although the PBWT supports efficient matching, storing the full PA requires substantial memory for large datasets such as the UK Biobank~\cite{halldorsson2022sequences}. To mitigate this, the $\mu$-PBWT~\cite{Cozzi2023-iv} was introduced, using run-length encoding (RLE) to reduce the space requirement from $O(\width\cdot \height)$ to $O(\tilde{r})$ words, where $\tilde{r}$ denotes the number of runs in the PBWT matrix. However, this compression comes at a cost:  steppings and $\phi$ queries are no longer constant-time, typically requiring predecessor or successor queries at each step~\cite{Cozzi2023-iv}. 

Recently, Bonizzoni et al.~\cite{Bonizzoni-fast-pbwt} showed that forward and backward stepping can be supported in constant time while preserving the $O(\tilde{r})$ space bound. Their approach splits PBWT runs into at most $2\tilde{r}$ sub-runs and leverages \emph{fractional cascading}~\cite{ChazelleG86}---a technique originally developed for computational geometry---to enable constant-time forward stepping and backward stepping between sites. This result motivates the search for an equally efficient solution for supporting $\phi$ queries in run-length compressed PBWT representations.

\subparagraph*{Related Work.}
The $\phi$ query was originally introduced for a text $T[1..n]$ to support the construction of the \emph{permuted longest common prefix} (PLCP) array~\cite{KarkkainenMP09}. For a text $T[1..n]$, it is defined as $\phi(\SA[j]) = \SA[j-1]$, where $\SA[1..n]$ denotes the suffix array of $T$.
Symmetrically, the query $\phi^{-1}$ is defined by $\phi(\SA[j]) = \SA[j+1]$.

By combining a classical FM-index~\cite{FerraginaM00} with a sampled suffix array, $\phi$ queries can be supported in $O(r + n/s)$ words of space and answered in $O(s \log \log n)$ time, where $s$ is a sampling parameter and $r$ is the number of runs in the Burrows--Wheeler Transform (BWT)~\cite{BurrowsWheeler1994} of $T$. 
The $\log \log n$ factor in the query time stems from the complexity of LF- and FL-mapping operations within the FM-index~\cite{FerraginaM00}.
This approach has been adopted to report all occurrences of a query pattern in a text~\cite{KarkkainenMP09}.

Gagie et al.~\cite{GagieNP20} revisited $\phi$-queries in the context of the $r$-index, reducing the space to $O(r)$ words and the query time to $O(\log \log_w (n/r))$. 
In this setting, $\phi$ queries are utilized to locate matches sequentially once an initial suffix array value has been identified, demonstrating that the efficiency of $\phi$ queries is a critical bottleneck in compressed text indexing.


Nishimoto and Tabei~\cite{NishimotoT21} introduced the \emph{move structure} to accelerate queries on the RLE-BWT. Brown et al.~\cite{BrownG022} subsequently established the \emph{the splitting theorem}, generalizing this concept to compactly represent permutations using move structures. Since both $\phi$ and $\phi^{-1}$ queries can be supported via permutations, the move structure reduces the total running time for $k$ iterative queries to $O(\log \log_w (n/r) + k)$ while maintaining $O(r)$-word space.

Cozzi et al.~\cite{Cozzi2023-iv} extended $\phi$-queries to the PBWT as defined above. Analogous to $\phi$-queries on texts, which are used to report all occurrences of a query pattern, $\phi$-queries in the PBWT are used to locate all SMEMs~\cite{Cozzi2023-iv}. 
They observe that if a site $j$ of a  haplotype $S_i$ is not a \emph{run-top}---that is, it shares the same symbol $S_i[j]$ as its co-lexicographic predecessor $S_{\phi_j(i)}$ at that site---then $i$ and $\phi_j(i)$ remain adjacent at site $j$ of the PA, i.e., $\phi_j(i) = \phi_{j+1}(i)$ (see Proposition~\ref{pro-consistent-color}).
Exploiting this observation, the supplementary material of~\cite{Cozzi2023-iv} briefly mentions that it is possible to use a successor data structure to answer $\phi$-queries.

For completeness, we outline the high-level idea of the approach by Cozzi et al.~\cite{Cozzi2023-iv}. 
For each haplotype, we identify its all sites at which its entries appear as run-tops in the PBWT.
Viewing the haplotype as an interval $[1,\width]$, we then partition it into \emph{haplotype intervals} at these sites.
Across all haplotypes, the total number of such intervals is $O(\tilde{r} + \height)$. 
The observation by Cozzi et al. implies that, for any haplotype interval $[b,e]$ of $S_i$, all values $\phi_j(i)$ are identical for $b \le j \le e$. Hence, it suffices to store $\phi_e(i)$ for each interval. 
To support queries, a successor data structure is built for each haplotype over the right endpoints of its intervals. Given a query $\phi_j(i)$, the structure identifies the interval of $S_i$ containing $j$ via a successor query and returns the value stored for that interval. As a result, one obtains an $O(\tilde{r} + \height)$-word data structure that supports each $\phi$-query in $O(\log\log_w \width)$ time. \lipicsEnd

As discussed earlier, $\phi$ queries are used to locate occurrences of query patterns in texts and, analogously, SMEMs in the PBWT. Reporting multiple occurrences requires a sequence of consecutive applications of $\phi$, i.e., $\phi^1, \phi^2, \dots, \phi^{k}$. We call this process \emph{iterative $\phi$ queries}.

\subparagraph*{Our Contributions.} In this work, we aim to support iterative $\phi$ queries more efficiently. To this end, we first provide a simple and novel algorithm that divides haplotype intervals into sub-intervals (referred to as \emph{refined segments} henceforth) such that
\begin{itemize}
    \item \emph{Overlap constraint:} each refined segment $[b, e]$ of haplotype $S_i$ overlaps with at most $d$ refined segments from haplotype $S_{i'}$, where $i' = \phi_e(i)$ and $d\ge 2$ is a constant, and
    \item \emph{Total-size constraint:}
    the total number of refined segments is bounded by $(\newR + \height)(1+\lceil \frac{1}{d-1} \rceil)$.
\end{itemize}

Our decomposition differs from the \emph{balancing algorithm}~\cite{NishimotoT21, BrownG022, Bertram0N24, nate-balancing-2026} used in prior work on \emph{move structures} ~\cite{NishimotoT21} and their generalization via the \emph{splitting theorem}~\cite{BrownG022}. Those approaches maintain a bijection between two lists of intervals, and each split may cause imbalances in previously processed intervals.
Namely, current operations may impact past computations.
In contrast, our decomposition handles $\height > 2$ lists of intervals. 
By traversing the PA column-by-column and then row-by-row, we ensure a sequential dependency where each operation affects only pending intervals and never invalidates refined segments already established.

This decomposition facilitates the design of two space--time tradeoffs for $k$ iterative $\phi$ queries.
The first maintains a space usage of $O((\newR + \height) \log n)$ bits while improving the query time from $O(k \cdot \log \log_w \width)$ to $O(\log \log_w (\min(\width, \height)) + k)$. The second reduces the space from $O((\newR + \height) \log n)$ to $O(\newR \log \height + \height \log n)$ bits, while supporting queries in $O(k \cdot \log \log_w \height)$ time.
%

In modern genomic datasets, such as the 1000 Genomes Project~\cite{10002015global}, the number  $\height$ is typically in the thousands, while $\width$ often reaches the millions.
Our second tradeoff is particularly effective when $\height \ll \width$. In such cases, it achieves notably better space efficiency and slightly improved time efficiency compared to the previous solution, which requires $O((\newR + \height)\log n))$ bits of space and supports queries in $(O(k \log \log_w \width))$ time.

\section{Preliminaries}
\label{sect-prel}

We assume the word-RAM model~\cite{FredmanW93} with word size $w = \Theta(\log n)$ bits, where $n = \width \cdot \height$. 


\subparagraph*{Notations.}
We denote by $[i,j]$ the interval of integers ${i, i+1, \dots, j}$, and define $[i,j] = \emptyset$ if $j < i$.
Given an interval $\mathsf{I} = [i,j]$, we denote its left and right endpoints by $\mathsf{I}.b$ and $\mathsf{I}.e$, respectively, so that $\mathsf{I}.b = i$ and $\mathsf{I}.e = j$.
We say that a list of intervals \emph{partitions} an interval $[b, e]$ if the intervals are pairwise disjoint and their union equals $[b, e]$.
For any matrix $A$, we denote by $\col_j(A)$ its $j$-th column and by $\col_j(A)[i]$ the entry $A[i][j]$.
Given two strings $\alpha$ and $\beta$, we say that $\alpha$ is \emph{co-lexicographically smaller} than $\beta$ if and only if (iff) one of the following holds: i) there exists an index $k$ such that $\alpha[|\alpha|-i+1]=\beta[|\beta|-i+1]$ for all $1\le i <k$, and $\alpha[|\alpha|-k+1]<\beta[|\beta|-k+1]$; ii) or $\alpha$ is a proper suffix of $\beta$.

\subparagraph*{Positional Burrows--Wheeler Transform (PBWT).}
Let $M$ be an $\height \times \width$ matrix that stores the haplotypes $S_1, S_2, \dots, S_\height$ in rows $1, 2, \dots, \height$. We assume that the input haplotypes are not necessarily pairwise distinct and that the character symbols within these haplotypes are drawn from a general alphabet.

Closely related to the PBWT is the \emph{Prefix Array (PA)}, a matrix that records, for each column, the permutation of haplotype indices induced by the PBWT.
Formally, the Prefix Array $\PA$ built for the matrix $M$ is an $\height\times \width$ matrix, in which $\col_1(\PA)$ is simply the list $1, 2, \dots, \height$, and $\col_j(\PA)$, for $j>1$, stores the permutation of the set $\{1, \dots, \height\}$ induced by the co-lexicographic ordering of prefixes of $\{S_1, \dots, S_\height\}$ up to column $j-1$, that is, $\col_j(\PA)[i] = k$ iff $S_k[1..j-1]$ is ranked $i$ in the co-lexicographic order of $S_1[1..j-1], \dots, S_{\height}[1..j-1]$. 

Let $\PBWT$ be the matrix representing the positional BWT of $M$. Then $\PBWT$ is also an $\height \times \width$ matrix in which $\col_j(\PBWT)[i] = \col_j(M)[\col_j(\PA)[i]]$ for all $i \in [1..\height]$ and $j \in [1..\width]$.
We refer to a maximal substring of identical characters in $\col_j(\PBWT)$ as a \emph{run}.
Throughout the paper, let $r_j$ denote the number of runs for $\col_{j}(\PBWT)$. We define $\newR$ as $\sum_{1\le j\le \width} r_j$.

To distinguish haplotype indices in the PA from the character symbols stored in the PBWT, we slightly abuse notation and refer to the entries of the PA as \emph{colors}.

\subparagraph*{The $\rank$ and $\select$ queries.}
Let $A[1..n]$ be an array of length $n$ over an alphabet $\Sigma$ of size $\sigma$. For any character $c \in \Sigma$, $\rank_c(A, i)$ returns the number of occurrences of $c$ in the prefix $A[1..i]$ for $1 \le i \le n$ , while $\select_c(A, j)$ returns the position of the $j$-th occurrence of $c$ in $A$ for $1 \le j \le \rank_c(A, n)$.

    

\begin{lemma}\label{lem-rank-select}\cite{BelazzouguiN15} 
There exists a data structure using $O(n\log \sigma)$ bits of space that supports $\rank_c(A, i)$ in $O(\log \log_w \sigma)$ time and $\select_c(A, j)$ in constant time.
\end{lemma}

\subparagraph*{Successor queries.}
Let $A \subseteq \{1, \dots, n\}$ be a set of $n'$ integers. For any $\alpha \in \{1, \dots, n\}$ and $1 \le i \le n'$, $\access(A, i)$ returns the $i$-th smallest element of $A$, and $\succe(A, \alpha)$ returns the smallest index $i$ such that $\access(A, i) \ge \alpha$, or $|A|+1$ if no such index exists.



\begin{lemma}\label{lem-pred-succ}\cite[Theorem~A.1]{BelazzouguiN15}\cite[Sparse bit vectors]{OkanoharaS07}
Given a set $A$ of $n'$ keys drawn from the universe $\{1, \dots, n\}$, there exists a data structure that represents $A$ using $O\!\left(n' \log \frac{n}{n'}\right)$ bits and supports $\access(A, i)$ in constant time and $\succe(A, \alpha)$ queries in $O\!\left(\log \log_w \frac{n}{n'}\right)$ time.
\end{lemma}

\subparagraph*{The definitions of $\phi$ and $\phi^{-1}$ in PBWT.}
Given a site $j \in [1,\width]$ and a haplotype index $i \in \{1,\dots,\height\}$, the $\phi_j(i)$ query returns the index (or color) of the haplotype that immediately precedes $i$ in $\col_j(\PA)$, while the $\phi^{-1}_j(i)$ query returns the index of the haplotype that immediately follows $i$ in $\col_j(\PA)$.
Formally, let $k \in [1,\height]$ be the row index such that $\col_j(\PA)[k] = i$. If $k = 1$, then $\phi_j(i) = 0$; otherwise, $\phi_j(i) = \col_j(\PA)[k-1]$. Symmetrically, if $k = \height$, then $\phi^{-1}_j(i) = \height + 1$; otherwise, $\phi^{-1}_j(i) = \col_j(\PA)[k+1]$. 
For example, in Figure~\ref{fig:hap-pbwt-pa}, we have $\col_4(\PA)=\{2,3,4,1,5\}$. In particular, $\col_4(\PA)[2] = 3$, $\phi_4(3) = 2$, and $\phi_4^{-1}(3) = 4$.

For any $j \in [1, \width]$ and $i \in [1, \height]$, define the iterated function $\phi^k_j(i)$ for $k \in \{1, 2, \dots, x-1\}$, where $x$ is the unique integer such that $\col_j(x) = i$, as follows: a) $\phi^1_j(i) = \phi_j(i)$, and $\phi^k_j(i) = \phi_j\big(\phi^{k-1}_j(i)\big)$ for $k > 1$. In the remainder of this paper, we present only the solution for iterative $\phi$ queries, as the one for iterative $\phi^{-1}$ follows symmetrically.


\subparagraph*{Run-tops and haplotype intervals.} For any $1\le i\le \height$ and $1\le j\le\width$, we refer to $(i, j)$ as a \emph{run-top} if $i=1$ or $i>1$ and $\col_j(\PBWT)[i]\ne \col_j(\PBWT)[i-1]$.
See Figure~\ref{fig:hap-pbwt-pa} for an example.


For each haplotype $S_i$, we partition the range $[1, \width]$ into a set of \emph{haplotype intervals}, denoted by $L_i$. The right endpoints of these intervals are defined by the sites where $S_i$ appears as a run-top in the PBWT. Formally, let $b_1 < b_2 < \dots < b_k = \width$ be the sorted columns such that for each $b_{\tau}\ne \width$, there exists a row index $i'$ where $(i', b_{\tau})$ is a run-top and $\col_{b_{\tau}}(\PA)[i'] = i$. These boundaries partition $[1, \width]$ into $k$ disjoint intervals $L_i = \{I_1, I_2, \dots, I_k\}$, where $I_1 = [1, b_1]$ and $I_{\tau} = [b_{\tau-1} + 1, b_{\tau}]$ for $1 < \tau \le k$.
As shown in Figure~\ref{fig:hap-inter-refined-seg}, the entries of haplotype $S_5$ (gray) that appear as run-tops in the PBWT are at sites $\{1, 2, 3, 5, 6\}$. Note that the fourth entry of $S_5$ does not appear as run-tops in the PBWT, though $S_5[4]\ne S_4[4]$. 
Following the definition of $L_i$, these site indices serve as the right endpoints for the partition of $[1, m]$. Accordingly, the haplotype intervals of $S_5$ are
$[1,1],[2,2], [3,3], [4, 5]$ and $[6,6]$.

\begin{figure}[t]
    \centering
    \includegraphics[width=1\textwidth]{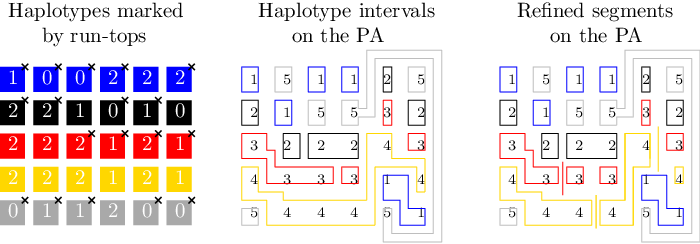}
   \caption{Example of haplotype intervals and refined segments. No refined segment overlaps more than $d = 2$ segments immediately above it in the PA.}
    \label{fig:hap-inter-refined-seg}
\end{figure}

Proposition~\ref{pro-consistent-color} ensures that haplotype intervals in $L_i$ define the maximal contiguous site ranges in the PA where the co-lexicographic predecessor of $i$ remains the same.

\begin{proposition}\label{pro-consistent-color}
Fix any $1 \le i \le \height$. Let $[b,e]$ be any haplotype interval in $L_i$. Then either for all $j \in [b,e]$ the colors $\phi_j(i)$ are all equal , or $\col(\PA)_j(1) = i$ holds for all $j \in [b,e]$.
\end{proposition}

\begin{proof}
For each $j \in [b,e]$, define $f(j)$ to be the index $x$ such that $\col_j(\PA)[x] = i$.
Since $[b,e]$ is a haplotype interval, the pair $(f(j), j)$ is not a run-top for any $j \in [b,e)$. 

We distinguish two cases based on the value of $f(b)$.
If $f(b) > 1$, since $(f(b), b)$ is not a run-top, it follows that $\col_b(\PBWT)[f(b)] = \col_b(\PBWT)[f(b)-1]$. By the stability of the sorting underlying $\PBWT$, this implies $\phi_b(i) = \phi_{b+1}(i)$. Repeating the same argument inductively for each $j \in [b,e)$, we obtain $\col_j(\PBWT)[f(j)] = \col_j(\PBWT)[f(j)-1]$ and $\phi_j(i) = \phi_{j+1}(i)$. Therefore, the colors $\phi_j(i)$ are identical for all $j \in [b,e]$.
Otherwise, $f(b) = 1$, so $(1,b)$ is a run-top by definition. Moreover, if $j > 1$, then $(1, j-1)$ is also a run-top, and if $j < \width$, then $(1, j+1)$ is a run-top as well. It follows that $b = e$. By the definition of $f(b)$, we have $\col_b(\PA)[1] = i$, and hence $\col_j(\PA)[1] = i$ for all $j \in [b,e]$.
\end{proof}





Proposition~\ref{pro-size} (proof in Appendix~\ref{app-prel}) bounds the total number of haplotype intervals.

\begin{proposition}\label{pro-size}
    The total number of haplotype intervals created lies in $[\newR, \newR + \height]$.
\end{proposition}



\section{Decomposition of Haplotype Intervals}
\label{sect-decomp}

In this section, we describe the algorithm that divides haplotype intervals into refined segments that satisfy both the overlap and the total-size constraints.


\subsection{The Decomposition Algorithm}



We assume the lists $L_1, \dots, L_{\height}$ of haplotype intervals are given, each implemented as a linked list. 
To construct the refined segments, we maintain an array $\HList[1..\height]$, where each $\HList[c]$ is a linked list of disjoint segments for haplotype $S_c$. These lists are initially empty.
The algorithm processes the Prefix Array ($\PA$) column by column from $j = 1$ to $\width$. Then, at each column $j$, we iterate through rows $i = 1$ to $\height$ and perform the following steps:

\begin{itemize}
    \item \textit{Retrieval:} retrieve the first interval $[b_c, e_c]$ stored in $L_c$, where $c=\col_j(\PA)[i]$.
    
    \item \textit{Passive Split:} if $j = e_c$, remove $[b_c, e_c]$ from $L_c$ and append it to the tail of $\HList[c]$.
    
    \item \textit{Active Split:} otherwise, if $i > 1$ and $[b_c, j]$ overlaps exactly $d$ refined segments in $\HList[c']$, where $d>1$ is a predefined constant and $c' = \col_j(\PA)[i-1]$, then append $[b_c, j]$ to the tail of $\HList[c]$, remove $[b_c, e_c]$ from $L_c$, and insert $[j+1, e_c]$ at the head of $L_c$.
\end{itemize}


An example of this decomposition process is shown in Figure~\ref{fig:the-alg}. Immediately before processing the fifth column of the PA, $L_1$ (blue) stores $\{[5,6]\}$; $L_2$ (black) and $L_3$ (red) store $\{[5,5],[6,6]\}$; $L_4$ (yellow) stores $\{[4,6]\}$; and $L_5$ (gray) stores $\{[4,5],[6,6]\}$.
At position $(2,5)$ of the PA, a Passive Split is triggered: the interval $[5,5]$ from $L_c$ (red), where $c=\col_5(\PA)[2]=3$, is removed and added to $\HList[c]$. Then, at position $(3,5)$, an Active Split is triggered, since $[4,5]$, a subinterval of $[4,6]$ in $L_c$ (yellow) with $c=\col_5(\PA)[3]=4$, overlaps $d(=2)$ refined segments in $\HList[c']$, where $c'=\col_5(\PA)[2]=3$ (red). As a result, $[4,5]$ is appended to $\HList[4]$, and $L_4$ is updated to $\{[6,6]\}$.
Finally, at position $(4,5)$, no split is triggered, so $L_c$ and $\HList[c]$ (blue), where $c=\col_5(\PA)[4]=1$, remain unchanged.

\begin{figure}[t]
    \centering
    \includegraphics[width=1\textwidth]{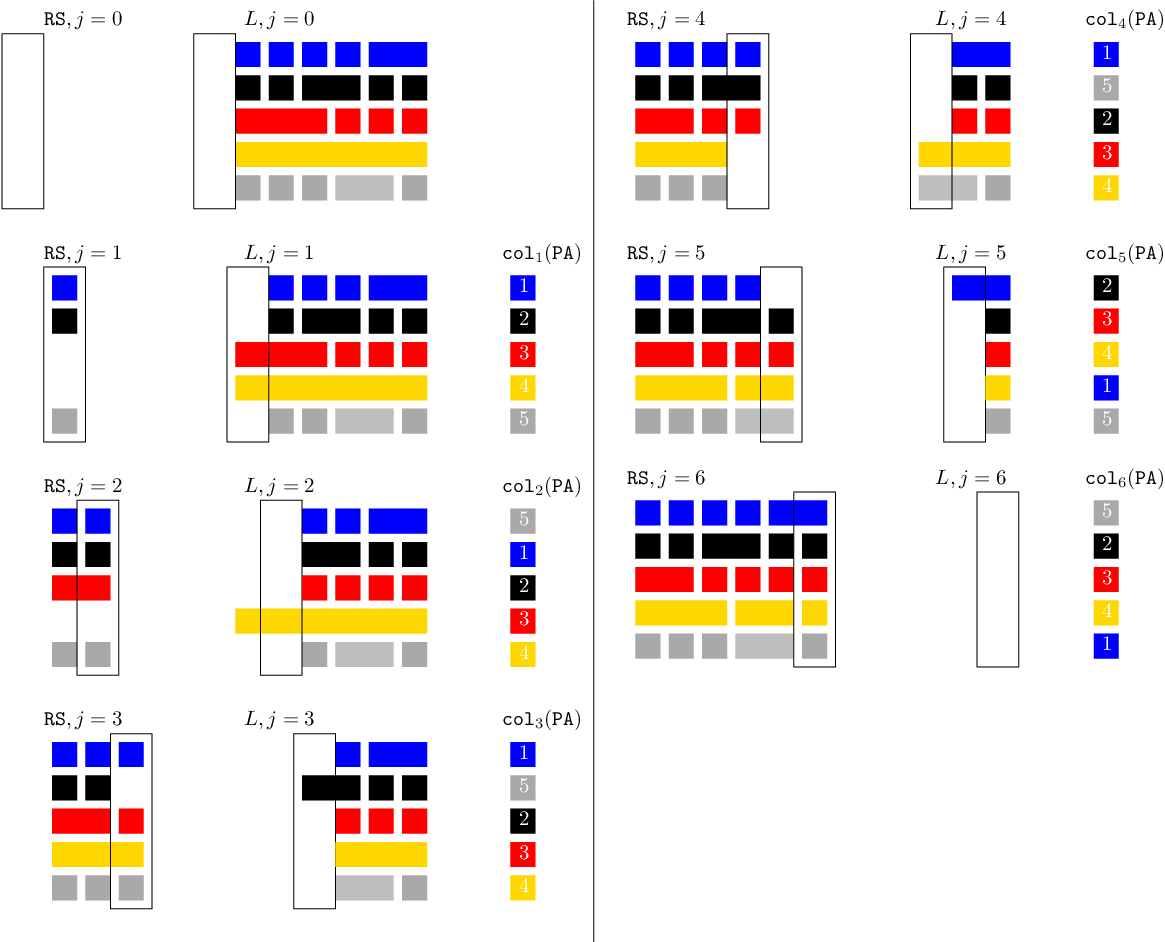}
\caption{Example of the states of $\HList$ and $L$ at column $j$, for $j \in [0, \width]$, with the parameter $d=2$.}
    \label{fig:the-alg}
\end{figure}



\subsection{The Correctness of the Algorithm}
We prove that after construction, the refined segments stored in $\HList$ satisfy both the overlap and total-size constraints.
All proofs omitted in the section can be found in Appendix~\ref{app-decomp}.

\subparagraph{Overview of Correctness.}
Let $c' = \phi_j(c)$. 
Recall that the haplotype index $c'$ is stored immediately above $c$ in column $j$ of the matrix PA. We process this matrix column by column (left to right) and rows in increasing order (1 to $\height$). This specific order ensures that when we reach an index $c$ at column $j$, all refined segments for haplotype $S_{c'}$ with right endpoints at most $j$ have already been established.

Let $[b_c, e_c]$ denote the interval retrieved from $L_c$ at a Retrieval step at column $j$. 
By the decomposition algorithm, we define a new refined segment $[b_c, j]$ for $S_c$ if $j$ is the right endpoint of an original haplotype interval, or if the interval $[b_c, j]$ overlaps exactly $d$ refined segments already established for $S_{c'}$. Lemma~\ref{lem-natural-term} shows that this construction ensures that the interval $[b_c, j]$ overlaps with at most $d$ refined segments of $S_{c'}$ at the moment of its creation.

When a refined segment $[b_c, j]$ is created for haplotype $S_c$, we prove that one of two exhaustive cases must occur: (i) The segment overlaps at most $d$ previously created refined segments of $S_{c'}$, and the most recent segment for $S_{c'}$ also terminates at $j$ (Lemma~\ref{lem-natural-term}); (ii) The segment overlaps fewer than $d$ such segments, and the most recent segment for $S_{c'}$ terminates strictly before $j$ (Lemma~\ref{lem-smaller-than-j}).
In the former, any refined segment for $S_{c'}$ created in a future step must begin after $j$, ensuring the overlap count for $[x_c+1, j]$ remains at most $d$. 
In the second case, exactly one future segment for $S_{c'}$ can overlap $[x_c+1, j]$, which increases the overlap count by at most one and thus preserves the overall upper bound of $d$.
Section~\ref{sect-size-constraint} shows a non-trivial amortized analysis bounding the total number of refined segments. 

\subsubsection{The Overlap Constraint}

We first prove that, after construction, $\HList[c]$ partitions $[1, \width]$ for every $c \in [1, \height]$.

\begin{lemma}\label{lem-basic}
Fix any color $1 \le c \le \height$. During the execution of the algorithm, if the left endpoint of the first interval in $L_c$ is $x$, then the refined segments stored in $\HList[c]$ partition the interval $[1, x-1]$, and the intervals stored in $L_c$ partition the interval $[x, \width]$.
\end{lemma}





\begin{corollary}\label{cor-basic}
After processing all $\width$ columns of the matrix $\PA$, the refined segments in $\HList[c]$, for each $1 \le c \le \height$, partitions the interval $[1,\width]$.
\end{corollary}

Next, we establish an upper bound on the number of refined segments in $\HList[\phi_j(c)]$ that overlap the interval $[b_c, j]$ used in either a Passive Split or an Active Split step. 


\begin{lemma}\label{lem-natural-term}
For any $c \in [1, \height]$, let $b_c$ be the left endpoint of the head interval of $L_c$ at column $j$. Let $A_j$ denote the refined segments currently in $\HList[c']$, where $c' = \phi_j(c)$. The interval $[b_c, j]$ overlaps at most $d$ refined segments in $A_j$. Furthermore, if $A_j$ is non-empty, then $b_c$ is contained within some refined segment in $A_j$.
\end{lemma}

\begin{proof}
We argue by contradiction. Let $A_j$ denote the state of $\HList[c']$ at column $j$. Suppose that $[b_c, j]$ overlaps more than $d$ refined segments in $A_j$. 
Let $j' < j$ be the leftmost column such that $[b_c, j']$ overlaps exactly $d$ refined segments in $A_j$. By the choice of $j'$, it must be the right endpoint of some refined segment in $A_j$. By the algorithm, this segment is appended to $\HList[c']$ at column $j'$. Let $A_{j'}$ denote the state of $\HList[c']$ at column $j'$.

By Lemma~\ref{lem-basic}, $\HList[c']$ partitions $[1, j']$ at column $j'$. Hence, $A_{j'}$ is a prefix of $A_j$, which implies that the interval $[b_c, j']$ overlaps exactly $d$ refined segments in $A_{j'}$ and that $b_c$ is contained in one of the refined segment in $A_{j'}$.

Since $b_c \le j' < j$, Proposition~\ref{pro-consistent-color} implies that $\phi_{j'}[c] = \phi_j[c] = c'$.  At column $j'$, the interval $[b_c, j']$ overlaps exactly $d$ refined segments in $A_{j'}$, and since $c' = \phi_{j'}[c]$, an Active Split step should have been triggered. Consequently, $[b_c, j']$ would have been appended to $\HList[c]$ as a refined segment, yielding a contradiction.
Hence, the assumption is false.
\end{proof}






By construction, any refined segment established at column $j$ terminates exactly at $j$. Consider a refined segment $[b_c, j]$ for haplotype $S_c$ created at site $j$. The right endpoint $e'$ of the most recently established segment in $\HList[c']$, where $c' = \phi_{j}(c)$, cannot exceed $j$, since the algorithm only processes columns up to $j$.
Lemma~\ref{lem-natural-term} (above) implies that if $e' = j$, then at most $d$ refined segments in $\HList[\phi_j(c)]$ overlap $[b_c, j]$. Lemma~\ref{lem-smaller-than-j} (below) shows that, if $e' < j$, then fewer than $d$ refined segments in $\HList[\phi_j(c)]$ overlap $[b_c, j]$.


\begin{lemma}\label{lem-smaller-than-j}
Let $[b_c, j]$ be a refined segment appended to $\HList[c]$ at column $j$, and let $c' = \phi_{j}(c)$. If the right endpoint of the most recently established segment in $\HList[c']$ is smaller than $j$, then $[b_c, j]$ overlaps strictly fewer than $d$ refined segments in $\HList[c']$.
\end{lemma}





\begin{proof}
If $b_c = j$, then $[b_c, j]$ overlaps at most one interval; since $d > 1$, the statement holds trivially.
Now suppose $j > b_c$. Let $A_j$ denote the state of $\HList[c']$ at column $j$. 
By Lemma~\ref{lem-natural-term}, $[b_c, j]$ overlaps at most $d$ refined segments in $A_j$.

Next, consider column $j-1$, and let $A_{j-1}$ denote the state of $\HList[c']$ at column $j-1$.
By the lemma's hypothesis, no refined segment is added to $\HList[c']$ at column $j$, which implies that $A_{j-1} = A_j$.
Since $A_j = A_{j-1}$ and $[b_c, j]$ overlaps at most $d$ refined segments in $A_j$, it follows that $[b_c, j-1]$ overlaps at most $d$ refined segments in $A_{j-1}$.

Because $[b_c, j]$ is a refined segment and all refined segments in $\HList[c]$ are pairwise disjoint (by Corollary~\ref{cor-basic}), the interval $[b_c, j-1]$ cannot itself be a refined segment. In particular, $j-1$ is not the right endpoint of any original haplotype interval in $L_c$. Since $b_c \le j-1 < j$, it follows that $c'=\phi_{j-1}(c)$ by Proposition~\ref{pro-consistent-color}.

If $[b_c, j-1]$ overlaps exactly $d$ segments in $A_{j-1}$, then an Active Split step would be triggered, and $[b_c, j-1]$ would become a refined segment, a contradiction. Hence, $[b_c, j-1]$ overlaps fewer than $d$ segments in $A_{j-1}$.

Finally, since $A_{j-1} = A_j$ and the right endpoint of the last interval in $\HList[c']$ is smaller than $j$, the intervals $[b_c, j-1]$ and $[b_c, j]$ overlap exactly the same set of segments in $A_j$. Therefore, $[b_c, j]$ overlaps fewer than $d$ refined segments in $A_j$, completing the proof.
\end{proof}

Lemma~\ref{lem-overlap-constraint} implies that the overlap constraint is satisfied. 


\begin{lemma}\label{lem-overlap-constraint}
Every refined segment $[b_c, j] \in \HList[c]$ satisfies the following properties:
(i) it is a sub-interval of some original haplotype interval in $L_c$, and
(ii) it overlaps at most $d$ refined segments in $\HList[c']$, where $c' = \phi_{j}(c)$, provided $c'> 0$.
\end{lemma}






\begin{proof}
Statement (i) follows directly from the algorithm. It remains to prove (ii).

Suppose that $c'>0$. Let $A_{j}$ denote the state of $\HList[c']$ at column $j$, and let $e'$ be the right endpoint of the last segment in $A_{j}$. By Lemma~\ref{lem-natural-term}, the segment $[b_c, j]$ overlaps at most $d$ segments in $A_{j}$. Moreover, by the algorithm, we have $e' \le j$. 

We distinguish two cases.
If $e' = j$, then any segment added to $\HList[c']$ after column $j$ has left endpoint greater than $j$ and therefore cannot overlap $[b_c, j]$. Hence, after processing all $\width$ columns, $[b_c, j]$ still overlaps at most $d$ segments in $\HList[c']$.
Otherwise, if $e' < j$, Lemma~\ref{lem-smaller-than-j} implies that $[b_c, j]$ overlaps fewer than $d$ segments in $A_{j}$. By the algorithm, the next segment added to $\HList[c']$ has right endpoint greater than $j$, and thus overlaps $[b_c, j]$, increasing the number of overlapping segments by one. Therefore, after processing all $\width$ columns, $[b_c, j]$ overlaps at most $d$ segments in $\HList[c']$.
Thus, the overlap constraint is satisfied.
\end{proof}

Before concluding this section, we show in Lemma~\ref{lem-d-j} that if a refined segment $[b_c, j]$ of haplotype $S_c$ overlaps exactly $d$ refined segments in $\HList[\phi_j(c)]$ at the time of its creation, then $j$ coincides with $e'$, the right endpoint of the last segment currently stored in $\HList[\phi_j(c)]$.

\begin{lemma}\label{lem-d-j}
When a refined segment $[b_c, j]$ is appended to the tail of $\HList[c]$ at column $j$, if $[b_c, j]$ overlaps $d$ refined segments in $\HList[c']$, where $c' = \phi_j[c]$, then the right endpoint of the last refined segment stored in $\HList[c']$ is $j$.
\end{lemma}

\subsubsection{The Total-Size Constraint}
\label{sect-size-constraint}

To bound the total number of refined segments, we henceforth call a refined segment \emph{canonical} if it is produced at an Active Split step;
otherwise, we call it non-canonical.
Thus, a canonical refined segment $[b_c, e_c]$ from $\HList[c]$ overlaps exactly $d$ refined segments from $\HList[c']$, where $c' = \phi_{e_c}(c)$.
For each canonical refined segment $[b_c, e_c]$, we (conceptually) define a set $E$ containing all $d$ refined segments from $\HList[c']$ that overlap $[b_c, e_c]$, and we call $E$ a \emph{canonical set}. Formally, a refined segment $[b_{c'}, e_{c'}]$ belongs to the canonical set of $[b_c, e_c]$ if it overlaps $[b_c, e_c]$ and $\phi_{e_c}(c) = c'$.
Proposition~\ref{pro-consistent-color} implies that $\phi_j(c) = c'$ for every $j \in [b_c, e_c]$. Therefore, all refined segments in the same canonical set come from the same haplotype, and hence are pairwise disjoint by Lemma~\ref{lem-basic}.
We say that a refined segment $[b_{c'}, e_{c'}]$ is \emph{relevant} to a canonical refined segment $[b_c, e_c]$ if $[b_{c'}, e_{c'}]$ belongs to the canonical set of $[b_c, e_c]$.

\begin{lemma}\label{lem-relevant}
Each refined segment is relevant to at most one canonical refined segment.
\end{lemma}

\begin{proof}
Assume, for the sake of contradiction, that there exists a refined segment $[x,y] \in \HList[c']$ that is relevant to two distinct canonical refined segments $[b_{c_1}, e_{c_1}]$ and $[b_{c_2}, e_{c_2}]$. Let $E_1$ and $E_2$ denote their corresponding canonical sets. Then $[x,y] \in E_1 \cap E_2$.

We first show that $x < b_{c_1}$ and $x < b_{c_2}$. We prove $x < b_{c_1}$; the argument for $x < b_{c_2}$ is identical. Since $[b_{c_1}, e_{c_1}]$ is canonical, it overlaps exactly $d$ refined segments in $\HList[c']$. By Lemma~\ref{lem-d-j}, the maximum right endpoint among all segments in $E_1$ is $e_{c_1}$. Because $[x,y] \in E_1$, we have $x \le y \le e_{c_1}$.
If $x \ge b_{c_1}$, then for every $j \in [x,y]$ we would have $\phi^{-1}_j(c') = c_1$. This would imply $c_1 = c_2$, and hence the canonical refined segments $[b_{c_1}, e_{c_1}]$ and $[b_{c_2}, e_{c_2}]$ overlap. By Lemma~\ref{lem-basic}, refined segments in $\HList[c']$ are pairwise disjoint, so these two segments must coincide, contradicting the assumption that they are distinct. Therefore, $x < b_{c_1}$.

Since $[x, y] \in E_1$ and $x < b_{c_1}$, it follows that $x$ is the minimum left endpoint of all segments in $E_1$. By the same logic, $x$ is also the minimum left endpoint of all segments in $E_2$.
Moreover, we have $x < b_{c_1} \le y$ and $x < b_{c_2} \le y$, and each of the canonical refined segments $[b_{c_1}, e_{c_1}]$ and $[b_{c_2}, e_{c_2}]$ overlaps exactly $d$ refined segments in $\HList[c']$. This implies that $E_1 = E_2$, and hence $e_{c_1} = e_{c_2}$.
If $c_1 = c_2$, then the two canonical refined segments coincide, contradicting the assumption that they are distinct. If $c_1 \ne c_2$ but $e_{c_1} = e_{c_2}$, then $\phi^{-1}_{e_{c_1}}(c') = c_1 \ne c_2 = \phi^{-1}_{e_{c_2}}(c')$, which is impossible.
In all cases we obtain a contradiction. 
Therefore, each refined segment is relevant to at most one canonical refined segment.
\end{proof}


Lemma~\ref{lem-total-number} establishes a relationship among the total count of refined segments, the number of canonical refined segments, and the total number of haplotype intervals.

\begin{lemma}\label{lem-total-number}
The total number of refined segments in $\HList$ equals the sum of the number of canonical refined segments and the total number of haplotype intervals. 
\end{lemma}




Lemma~\ref{lem-total-canonical} bounds the number of canonical segments; together with Lemma~\ref{lem-total-number}, this establishes the upper bound on the total number of refined segments.


\begin{lemma}\label{lem-total-canonical}
An Active Split step is triggered at most $\left\lceil \frac{Y}{d-1} \right\rceil$ times, so, at most $\left\lceil \frac{Y}{d-1} \right\rceil$ canonical refined segments are created, where $Y$ is the total number of haplotype intervals.
\end{lemma}

\begin{proof}
Whenever a canonical refined segment $[b,e]$ is produced at an Active Split step, we conceptually mark the $d$ refined segments in the canonical set associated with $[b,e]$. By Lemma~\ref{lem-relevant}, each refined segment can be marked at most once.

Let $\alpha$ and $\beta$ denote the current numbers of unmarked refined segments and remaining intervals in $L$, respectively, during the execution of the algorithm. Initially, $\alpha = 0$ and $\beta = Y$.

At each step, if Passive Split occurs, then one interval is moved from $L$ to $\HList$ as an unmarked refined segment. Thus, $\alpha$ increases by one and $\beta$ decreases by one.
If an Active Split occurs, then $d$ refined segments are marked and one new unmarked refined segment is added to $\HList$. Hence, $\alpha$ decreases by $d-1$, while $\beta$ remains unchanged.
Therefore, throughout the execution, $\alpha$ never exceeds $Y$. Since each Active Split step reduces $\alpha$ by $d-1$, it can be triggered at most $\left\lceil \frac{Y}{d-1} \right\rceil$ times.
\end{proof}

Theorem~\ref{theorem-dec} summarizes the properties of our decomposition algorithm.

\begin{theorem}\label{theorem-dec}
There exists an algorithm that divides the haplotype intervals in $L$ into refined segments stored in $\HList$ such that: (i) each refined segment is a subinterval of some haplotype interval; (ii) for every $1 \le c \le \height$, the refined segments in $\HList[c]$ form a partition of $[1,\width]$; (iii) each refined segment $[b,e]$ in $\HList[c]$ overlaps with at most $d$ refined segments in $\HList[c']$, where $c' = \phi_e(c)$, (whenever $c' > 0$); and (iv) the total number of refined segments is at most $\newR + \height + \left\lceil \frac{\newR + \height}{d-1} \right\rceil$, for any constant $d > 1$.
\end{theorem}




\section{Two Space-Time Tradeoffs for $\phi$ Queries}
\label{sect-tradeoff}

We present two space-time tradeoffs for iterative $\phi$ queries, both based on the decomposition of haplotype intervals given by Theorem~\ref{theorem-dec}.
We assume, in this section, that for all $1 \le i \le \height$, the lists $\HList[i]$ of refined segments---constructed according to Theorem~\ref{theorem-dec} with $d := 2$---are available.
The two tradeoffs are summarized as follows:

\begin{theorem}
There exists a data structure using $S(n, \newR, \height, \width)$ bits of space that supports $\phi^1, \phi^2, \dots, \phi^k $ queries on $\PBWT$ in overall $Q(n, \newR, \height, \width, k)$ time, where 
\begin{itemize}
    \item i) $S(n, \newR, \height, \width)=O((\newR+\height)\log n)$ and $Q(n, \newR, \height, \width, k)=O(\log \log_w \min(\width,\height)+k)$, or 
    \item ii) $S(n, \newR, \height, \width)=O(\newR \log \height + \height \log n)$ and $Q(n, \newR, \height, \width,k)=O(k \log \log_w \height)$.
\end{itemize}
\end{theorem}

Henceforth, for any $1 \le c \le \height$ and any refined segment $[b, e] \in \HList[c]$, we call $c$ the \emph{color} of $[b, e]$. 
If $\phi_e(c)>0$, we define the \emph{parent} of $[b, e]$ to be the refined segment in $\HList[\phi_e(c)]$ that contains $e$. 
Since the refined segments in $\HList[\phi_e(c)]$ partition $[1, \width]$ by Corollary~\ref{cor-basic}, this parent segment is unique if it exists.
As shown in Figure~\ref{fig:hap-inter-refined-seg}, the color of the refined segment $[4, 5]$ from $S_5$ (in gray color) is $5$. Its parent is the refined segment $[5, 6]$ from $S_1$, since $\phi_5(5)=1$.

\subparagraph*{High-Level Idea.}
Inspired by the move structure~\cite{NishimotoT21}, we eliminate expensive successor queries by maintaining, for each refined segment, its relationship to its parent: the color $c'$, the endpoints $[b_{c'}, e_{c'}]$, and the index of the parent in $\HList[c']$. By the overlap constraint, it suffices at each iteration to scan only a constant number of refined segments (e.g., in $\HList[c']$) to find the one containing $j$, assuming the queries are applied at site $j$.
In the first trade-off, we store this relationship information explicitly, using $\Theta(\log n)$ bits per segment: $\lceil \log \width \rceil$ bits for the endpoints and $\lceil \log \height \rceil$ bits for the color. This yields constant time per step after an initial successor search.
In the second trade-off, we store the endpoints in a sparse bit vector (Lemma~\ref{lem-pred-succ}), reducing the space to $O(\log \height)$ bits per segment. To recover the parent index, we encode the relative positions of right endpoints of all refined segments in a single array over $\{1,\dots,\height\}$, using $O(\log \height)$ bits per segment, and support rank and select queries via Lemma~\ref{lem-rank-select}. This introduces a rank query at each step, increasing the time per step to $O(\log \log_w \height)$.

\subsection{The First Tradeoff}


\subparagraph*{The Data Structures.}
We construct an array $X[1,\height]$ of $\height$ integers such that $X[i]=\sum_{1\le i'< i} |\HList[i']|$ for $1\le i\le \height$, so in particular $X[1]=0$.

Let $\alpha=\sum_{1\le i\le \height} |\HList[i]|$. Equivalently, $\alpha=X[\height]+|\HList[\height]|$ is the total number of refined segments. Let $\HList[i][\ell]$, for $1\le \ell\le |\HList[i]|$, denote the $\ell$-th refined segment stored in $\HList[i]$.

Then, we construct an array $T[1..\alpha]$ of triples. For each refined segment, $T$ stores its right endpoint $e$, the index $\iota$ of its parent refined segment, and the color $c$ of that parent segment in $\HList[c]$.
To this end, we iterate $i$ from $1$ to $\height$, and for each $i$, we iterate $\ell$ from $1$ to $|\HList[i]|$. For each pair $(i,\ell)$, we add a triple $(e,c,\iota)$ to $T[X[i]+\ell]$ such that: 
\begin{itemize}
    \item $T[X[i]+\ell].e\in [1, \width]$, is the right endpoint of the refined segment $\HList[i][\ell]$,
    \item $T[X[i]+\ell].c\in [1,\height]$, is the color of its parent refined segment, i.e., $\phi_{E[X[i]+\ell].e}(i)$, and
    \item $T[X[i]+\ell].\iota$, drawn from $\{1, \dots, |\HList[T[X[i]+\ell].c]|\}$, is the index of the refined segment in $\HList[T[X[i]+\ell].c]$ that contains $T[X[i]+\ell].e$.
\end{itemize}

Next, depending on whether $\width < \height$, we construct different data structures to store the right endpoints of refined segments and support access and successor queries.

If $\width < \height$, let $E_i[1..|\HList[i]|]$ be the array storing the right endpoints, drawn from $\{1, \dots, \width\}$, of the refined segments in $\HList[i]$, for $1 \le i \le \height$. We build the data structure of Lemma~\ref{lem-pred-succ} over each array $E_i$, separately, to support successor queries. Additionally, we store a list of $\height$ pointers, each using $\Theta(\log n)$ bits, to access the data structure associated with any $E_i$.

Otherwise, we instead construct an array $P[1..\alpha]$ that collectively stores these endpoints over the universe $\{1,\dots,n\}$. 
To build $P$, we iterate through $i=1,\dots,\height$ and, for each $i$, through $\ell=1,\dots,|\HList[i]|$. For each pair $(i,\ell)$, we set $P[X[i]+\ell] = \width \cdot (i-1) + e$, where $e\in \{1, \dots, \width\}$ is the right endpoint of the refined segment $\HList[i][\ell]$. After constructing $P$, we build the data structure of Lemma~\ref{lem-pred-succ} over $P[1..\alpha]$ to support successor and access queries, after which $P$ can be discarded.

Storing the array $X$ requires $O(\height \log \alpha) \subseteq O(\height \log n)$ bits, since the total number of refined segments is $\alpha$ and $\alpha \le n$. Each triple in $T[1..\alpha]$ stores a triple of $(\log \width + \log \height + \log \width)$ bits, for a total of $O(\alpha\log n)$ bits, since $\log \width + \log \height = \log n$.
If $\width < \height$, then the successor-query data structures of Lemma~\ref{lem-pred-succ} built over all arrays $E_i$, together with the $\height$ pointers, occupy $O(\alpha\log \width + \height \log n) \subseteq O(\alpha \log n)$ bits of space.
Otherwise, the data structure of Lemma~\ref{lem-pred-succ} built over $P[1..\alpha]$ requires $O(\alpha \log (n/\alpha))$ bits, since it stores $\alpha$ elements drawn from $\{1, \dots, n\}$.
The total space usage is $O(\alpha \log n) \subseteq O((\newR + \height)\log n)$ bits, since $\alpha\le \tilde{r}+\height$ by Proposition~\ref{pro-size}.

\subparagraph*{The Query Algorithm.}
Consider the iterative queries $\phi^1_j(i), \phi^2_j(i), \dots, \phi^k_j(i)$. 
To support them, for $0 \le \kappa \le k$, let $c_{\kappa} = \phi^{\kappa}_j(i)$ and let $\ell_{\kappa}$ denote the index of the refined segment in $\HList[c_{\kappa}]$ that contains $j$. 
For convenience, define $\phi^0_j(i) = i$, so $c_0 = i$.

Let $0 \le \kappa < k$ and assume that both $c_{\kappa}$ and $\ell_{\kappa}$ are known. 
We show how to compute $c_{\kappa+1}$ and $\ell_{\kappa+1}$. 
By the definitions of $T[1..\alpha]$ and $X[1..\height]$, the entry $T[X[c_{\kappa}] + \ell_{\kappa}]$ stores the triple corresponding to the refined segment $[b,e]$ in $\HList[c_{\kappa}]$ that contains $j$. 
Thus, $c_{\kappa+1}$ is simply given by $T[X[c_{\kappa}] + \ell_{\kappa}].c$. 

Recall that the parameter $d$ is set to $2$, so the segment $[b,e]$ overlaps with at most two refined segments in $\HList[c_{\kappa+1}]$. 
Since $j \in [b,e]$, the refined segment in $\HList[c_{\kappa+1}]$ that contains $j$ must be either the $T[X[c_{\kappa}] + \ell_{\kappa}].\iota$-th or the $(T[X[c_{\kappa}] + \ell_{\kappa}].\iota - 1)$-th segment. 
Hence, both $c_{\kappa+1}$ and $\ell_{\kappa+1}$ can be determined in constant time.

It remains to compute $\ell_0$, i.e., the index of the refined segment $[b_0,e_0]$ in $\HList[i]$ that contains $j$. We distinguish two cases.
If $\width < \height$, observe that $e_0$ is the successor of $j$ in $E_i$. Hence, a successor query using Lemma~\ref{lem-pred-succ} returns the rank $\ell_0$ in $O(\log \log_w \width)\subseteq O(\log \log_w \min(\height, \width))$ time, since every entry in $E_i$ lies in $\{1, \dots, \width\}$ and $\width<\height$.
Otherwise, note that in $P[1..\alpha]$, all right endpoints of segments in $\HList[i]$ are encoded as integers in the range $[\width\cdot(i-1)+1, \width\cdot i]$. We therefore compute $\ell_0$ as the successor rank $\succe(P, \width(i-1)+j)$ using Lemma~\ref{lem-pred-succ}. Correctness follows because $\width(i-1)+e_0$ is the successor of $\width(i-1)+j$ in $P$. 

By Lemma~\ref{lem-pred-succ}, $\ell_0$ can be found in $O(\log \log_w (n/\alpha)) \subseteq O(\log \log_w \height)\subseteq O(\log \log_w \min(\height, \width))$ time, since $\width\ge \height$, $n/\alpha \le n/\tilde{r} \le \height$ and $\tilde{r} \ge \width$.
Once $\ell_0$ is known, we can compute $c_{\kappa}$ and $\ell_{\kappa}$ for $1 \le \kappa \le k$ sequentially, each in constant time. 
Thus, the overall running time is $O(\log \log_w \min(\height, \width) + k)$.

\subsection{The Second Tradeoff}

To achieve the second tradeoff, we must eliminate all entries requiring $\lceil \log \width \rceil$ bits of space. Specifically, we can no longer explicitly store: (i) the right endpoint $e$ and parent index $\iota$ within the triples of $T[1, \alpha]$, and (ii) the auxiliary data structures $E_i[1..|\HList[i]|]$ used for retrieving $\ell_0$ for each haplotype $1 \le i \le \height$.

\subparagraph{The Data Structures.} Instead, we maintain the array $X[1..\height]$.
In the array $T[1..\alpha]$, instead of storing the triples, we retain only the entry $c$ from each triple, representing the color of a parent refined segment.
We keep the data structure constructed over the array $P[1..\alpha]$ that is supposed to be built for the first tradeoff, when $\width\ge\height$.

To recover the index of a parent refined segment (a role originally served by the entry $\iota$ of each triple in the first tradeoff), we introduce an auxiliary data structure $I[1..\alpha]$.
 The array $I$ encodes the relative ordering of the right endpoints of all refined segments with respect to the original haplotypes $\{S_1, \dots, S_\height\}$.


To construct $I[1..\alpha]$, we initialize it as an empty list. We then iterate through $j=1,\dots,\width$, and for each $j$, through $i=1,\dots,\height$. For each pair $(i,j)$, if $j$ is the right endpoint of some refined segment in $\HList[i]$, we append $i$ to $I$. 
In the example shown in Figure~\ref{fig:hap-inter-refined-seg}, the array $I[1..\alpha]$ is $\{1,2,5,1,2,3,5,1,3,4,5,1,2,3,2,3,4,5,1,2,3,4,5\}$.

After constructing $I$, we build the data structure of Lemma~\ref{lem-rank-select} on $I[1..\alpha]$ to support $\rank$ and $\select$ queries, after which $I$ can be discarded.
The data structure of Lemma~\ref{lem-rank-select} supports $\rank$ and $\select$ queries in $O(\log \log_w \height)$ and $O(1)$ time, respectively, since all elements are drawn from a universe of size $\height$.

The space usage is as follows. Storing $X$ requires $O(\height \log n)$ bits. Each entry of $T[1..\alpha]$ now requires $\log \height$ bits, for a total of $O(\alpha \log \height)$ bits. Since $P$ contains $\alpha$ values from the universe $\{1,\dots,n\}$, the structure of Lemma~\ref{lem-pred-succ} requires $O(\alpha \log \frac{n}{\alpha}) \subseteq O(\alpha \log \height)$ bits, as $\alpha \ge \newR \ge \width$ (Proposition~\ref{pro-size}) and $n=\width\cdot\height$. Each entry of $I$ requires $O(\log \height)$ bits, and the data structure of Lemma~\ref{lem-rank-select} built on $I[1..\alpha]$ occupies $O(\alpha \log \height)$ bits. Overall, the total space usage is $O(\alpha \log \height + \height \log n)\subseteq O(\tilde{r}\log \height+\height\log n)$ bits (Proposition~\ref{pro-size}).

\subparagraph*{The Query Algorithm.}
Consider iterative queries $\phi^1_j(i), \phi^2_j(i), \dots, \phi^k_j(i)$. 
Recall that $c_0 = i$

To find $\ell_0$, i.e., the index of the refined segment $[b_{c_0}, e_{c_0}]$ in $\HList[i]$ that contains $j$, we apply the algorithm from the first tradeoff, which is used when $\width \ge \height$. Specifically, we compute $\ell_0$ as the successor rank $\succe(P, \width(i-1)+j)$ using Lemma~\ref{lem-pred-succ}. By Lemma~\ref{lem-pred-succ}, $\ell_0$ can be computed in $O(\log\log_w n/\alpha) \subseteq O(\log\log_w \height)$ time.

It remains to show how to compute $c_{\kappa+1}$ and $\ell_{\kappa+1}$ for $0 \le \kappa < k$, assuming that $c_\kappa$ and $\ell_\kappa$ are already known. The computation of $c_{\kappa+1}$ is immediate: by the definitions of $T[1..\alpha]$ and $X[1..\height]$, we have $c_{\kappa+1} = T[X[c_\kappa] + \ell_\kappa]$, which can be evaluated in $O(1)$ time.

By definition, $\ell_{\kappa+1}$ is the index of the refined segment $[b_{c_{\kappa+1}}, e_{c_{\kappa+1}}]$ in $\HList[c_{\kappa+1}]$ that contains $j$. Let $[b_{c_\kappa}, e_{c_\kappa}]$ denote the $\ell_\kappa$-th refined segment in $\HList[c_\kappa]$, i.e., $\HList[c_\kappa][\ell_\kappa]$.

To compute $\ell_{\kappa+1}$, the algorithm in the first tradeoff finds the index $\ell$ of the refined segment in $\HList[c_{\kappa+1}]$ that contains $e_{c_\kappa}$. Since $\ell$ is explicitly stored in the triple $T[X[c_\kappa] + \ell_\kappa]$ as the entry $\iota$, it can be retrieved in constant time. In the second tradeoff, however, we remove $\iota$ from each triple to save space, and instead use the data structures built on $P[1..\alpha]$ and $I[1..\alpha]$ to recover $\ell$.
Specifically, we locate the position $x$ of the $\ell_\kappa$-th occurrence of $c_\kappa$ in $I[1..\alpha]$, i.e., $x = \select_{c_\kappa}(I, \ell_\kappa)$. By Lemma~\ref{lem-correct} (deferred to Appendix~\ref{app-correct-second}), either the $\rank_{c_{\kappa+1}}(I, x)$-th or the $(1+\rank_{c_{\kappa+1}}(I, x))$-th occurrence of $c_{\kappa+1}$ in $I$ corresponds to the refined segment in $\HList[c_{\kappa+1}]$ that contains $e_{c_{\kappa}}$.
After computing $x$, we retrieve the right endpoints of these two candidate segments via $\access$ queries on $P[1..\alpha]$, in particular obtaining the segment $\HList[c_{\kappa+1}][1+\rank_{c_{\kappa+1}}(I, x)]$. If this segment contains $e_{c_\kappa}$, we set $\ell := 1+\rank_{c_{\kappa+1}}(I, x)$; otherwise, we set $\ell := \rank_{c_{\kappa+1}}(I, x)$. Once $\ell$ is determined, we apply the same procedure as in the first tradeoff to compute $\ell_{\kappa+1}$ from $\ell$.

The computation of $\ell_{\kappa+1}$ involves one $\select$ query, two $\rank$ queries on $I[1..\alpha]$, and two $\access$ queries on $P[1..\alpha]$, yielding a running time of $O(\log \log_{w} \height)$.
In summary, $\ell_0$ can be computed in $O(\log\log_w \height)$ time, and each subsequent pair $(c_\kappa, \ell_\kappa)$ can be obtained in $O(\log \log_{w} \height)$ time. Therefore, the total running time for $k$ iterative queries is $O(k \cdot \log \log_{w} \height)$.

\bibliography{phi-main}

\appendix

\section{Proofs Omitted in Section~\ref{sect-prel}}
\label{app-prel}

\subsection{The Proof of Proposition~\ref{pro-size}}

\begin{proof}
Let $B_i$ be the set of index--endpoint pairs $(i,y)$, where $y$ is the right endpoint of a haplotype interval in $L_i$, for each $1 \le i \le \height$.

Each run induces exactly one run-top, and each run-top contributes a distinct pair to $\bigcup_i B_i$. Moreover, $(i,\width) \in B_i$ for every $i \in [1,\height]$. Hence, $\tilde{r} \le \sum_i |B_i| \le \newR + \height$.

Since the total number of haplotype intervals equals $\sum_i |B_i|$, it follows that this number lies between $\tilde{r}$ and $\newR + \height$.
\end{proof}




\section{Proof Omitted in Section~\ref{sect-decomp}}
\label{app-decomp}

\subsection{The Proof of Lemma~\ref{lem-basic}}

\begin{proof}
The proof proceeds by induction. In the base case, before processing the first column of the matrix $\PA$, the list $\HList[c]$ is empty, and the intervals in $L_c$ are the original haplotype intervals from haplotype $S_c$, which partition $[1, \width]$. Hence, the base case holds trivially.

For the inductive step, assume that immediately after processing column $j-1$, the list $\HList[c]$ partitions $[1, x-1]$ and the list $L_c$ partitions $[x, \width]$. By the algorithm, we have $x \le j$.

While processing $L_c$ at column $j$, the Retrieval step selects the first interval $[x, y]$ in $L_c$, where $y \ge x$. According to the algorithm, after processing $L_c$ and $\HList[c]$, either both lists remain unchanged, or $\HList[c]$ partitions $[1, y]$ and $L_c$ partitions $[y+1, \width]$.

In either case, the inductive hypothesis continues to hold after processing column $j$, which completes the proof.
\end{proof}

\subsection{The Proof of Lemma \ref{lem-d-j}}

\begin{proof}
Since $d > 1$, we have $j > b_c$. Because $[b_c, j]$ is a refined segment and all refined segments in $\HList[c]$ are pairwise disjoint (by Corollary~\ref{cor-basic}), the interval $[b_c, j-1]$ cannot itself be a refined segment. Hence, $j-1$ is not the right endpoint of any original haplotype interval in $L_c$. Since $b_c \le j-1 < j$, we have $c' = \phi_{j-1}(c)$ by Proposition~\ref{pro-consistent-color}.

Let $A_j$ and $A_{j-1}$ denote the state of $\HList[c']$ at columns $j$ and $j-1$, respectively. Clearly, $A_{j-1}$ can be regarded as a prefix of $A_j$, and thus $A_{j-1} \subseteq A_j$. Since $[b_c, j-1] \subseteq [b_c, j]$ and $[b_c, j]$ overlaps $d$ intervals in $A_j$, $[b_c, j-1]$ overlaps at most $d$ intervals in $A_{j-1}$.

If $[b_c, j-1]$ overlaps exactly $d$ intervals in $A_{j-1}$, then an Active Split step would be triggered, and $[b_c, j-1]$ would become a refined segment, a contradiction. Therefore, $[b_c, j-1]$ overlaps fewer than $d$ intervals in $A_{j-1}$.

Finally, since $[b_c, j]$ overlaps $d$ refined segments in $A_j$, while $[b_c, j-1]$ overlaps fewer than $d$ refined segments in $A_{j-1}$, it follows that a new refined segment must be appended to $\HList[c']$ at column $j$. By either a Passive Split or a Active Split step, this newly added segment has right endpoint $j$, which completes the proof.
\end{proof}

\subsection{The Proof of Lemma~\ref{lem-total-number}}

\begin{proof}
Let $A$ (resp. $B$) be the set of right endpoints of non-canonical refined segments (resp. haplotype intervals). Recall that a non-canonical refined segment is produced at a Passive Split step. Its right endpoint coincides with the endpoint of some haplotype interval, which establishes a one-to-one correspondence between the elements of $A$ and $B$. Hence, $|A| = |B|$.

Observe that the right endpoint of any canonical refined segment does not belong to $B$ and that the total number of refined segments is the sum of the number of canonical refined segments and the number of non-canonical refined segments. Since the latter is $|A|$, which also equals $|B|$, the claim follows.
\end{proof}

\subsection{The Proof of Theorem~\ref{theorem-dec}}

\begin{proof}
Statements (i) and (ii) follow from Lemma~\ref{lem-basic}, while (iii) follows from Lemma~\ref{lem-overlap-constraint}.

Let $Y$ denote the total number of haplotype intervals. By Lemmas~\ref{lem-total-number} and~\ref{lem-total-canonical}, the total number of refined segments is bounded by $Y + \left\lceil \frac{Y}{d-1} \right\rceil$. Since $Y \le \newR + \height$ by Proposition~\ref{pro-size}, statement (iv) follows.
\end{proof}

\section{The Details Omitted in Section~\ref{sect-tradeoff}}
\label{app-correct-second}

\subsection{The Correctness of the Algorithm in the Second Tradeoff}

The correctness follows from Lemma~\ref{lem-correct}.

\begin{lemma}\label{lem-correct}
Let $x$ be the position of the $\ell_{\kappa}$-th occurrence of $c_{\kappa}$ in $I[1..\alpha]$. Then either the $\rank_{c_{\kappa+1}}(I, x)$-th or the $(1+\rank_{c_{\kappa+1}}(I, x))$-th refined segment in $\HList[c_{\kappa+1}]$ contains $e_{c_{\kappa}}$, i.e., the right endpoint of the $\ell_{\kappa}$-th refined segment in $\HList[c_{\kappa}]$.
\end{lemma}

\begin{proof}
Let $v < x$ (resp., $z > x$) be the rightmost (resp., leftmost) position in $I$ where the color $c_{\kappa+1}$ occurs, i.e., $v = \select_{c_{\kappa+1}}(I, \rank_{c_{\kappa+1}}(I, x))$ (resp., $z = \select_{c_{\kappa+1}}(I, 1+\rank_{c_{\kappa+1}}(I, x))$). At least one of $v$ or $z$ must exist.

Observe that $I[v]$ and $I[z]$ correspond to the right endpoints of the $\rank_{c_{\kappa+1}}(I, x)$-th and the $(1+\rank_{c_{\kappa+1}}(I, x))$-th refined segments in $\HList[c_{\kappa+1}]$, respectively. Hence, it suffices to show that, during the construction of $I[1..\alpha]$, either $I[v]$ or $I[z]$ is inserted due to the right endpoint $e_{c_{\kappa+1}}$. Recall that $e_{c_{\kappa+1}}$ is the right endpoint of the refined segment in $\HList[c_{\kappa+1}]$ that contains $e_{c_{\kappa}}$, so it follows that $e_{c_{\kappa+1}} \ge e_{c_{\kappa}}$.

We distinguish two cases. First, if $c_{\kappa+1} > c_{\kappa}$, then by the construction of $I[1..\alpha]$, the occurrence $c_{\kappa+1}$ corresponding to the endpoint $e_{c_{\kappa+1}}$ is appended to $I$ later than the occurrence $c_{\kappa}$ corresponding to the endpoint $e_{c_{\kappa}}$. Hence, $I[z]$, which is greater than $x$, corresponds to the endpoint $e_{c_{\kappa+1}}$.

Otherwise, $c_{\kappa+1} < c_{\kappa}$ (note that equality is impossible). If $e_{c_{\kappa+1}} > e_{c_{\kappa}}$, then again $e_{c_{\kappa+1}}$ corresponds to $I[z]$ by construction. Otherwise, we have $e_{c_{\kappa+1}} = e_{c_{\kappa}}$. In this case, since $c_{\kappa+1} < c_{\kappa}$ and the endpoints coincide, the occurrence $c_{\kappa+1}$ is appended to $I$ earlier than the occurrence $c_{\kappa}$ corresponding to the endpoint $e_{c_{\kappa}}$. Thus, $I[v]$, which is less than $x$, corresponds to the endpoint $e_{c_{\kappa+1}}$.

In all cases, either $I[v]$ or $I[z]$ is inserted into $I$ due to the right endpoint $e_{c_{\kappa+1}}$, which proves the claim.
\end{proof}

Consider the example shown in Figure~\ref{fig:hap-inter-refined-seg}. 
The refined segment $[4, 5]$ in $S_5$ corresponds to the fourth occurrence of $5$ in $I[1..\alpha]$. Its parent refined segment is $[5,6]$ in $S_1$, with color $1$, corresponding to the occurrence of $1$ in $I$ following the fourth occurrence of $5$.

\end{document}